%% file: ms.tex
\documentclass[sigplan]{acmart}
\usepackage{enumitem}
\usepackage{listings}
\usepackage{pdfcomment}
\usepackage[colorinlistoftodos,shadow,textsize=tiny,textwidth=1.5cm]{todonotes}
\usepackage{marginnote}
\let\marginpar\marginnote

\newcommand{\yibo}[2]{\textcolor{blue}{#1}\todo[fancyline,color=blue!40]{\textbf{Yibo}: #2}}

\acmConference[HPDC'22]{ACM Symposium on High-Performance Parallel and Distributed Computing}{June 27--July 1, 2022}{Minneapolis, Minnesota, USA}
\acmYear{}
\acmISBN{} 
\acmDOI{} 
\startPage{1}

\setcopyright{none}

\usepackage{booktabs}   
\usepackage{subcaption} 

\usepackage[ruled,linesnumbered,lined,boxed,commentsnumbered]{algorithm2e}

\usepackage{epigraph}

\usepackage{etoolbox}
\makeatletter
\patchcmd{\epigraph}{\@epitext{#1}}{\itshape\@epitext{#1}}{}{}
\makeatother

\usepackage{xcolor}
\usepackage{pgfplots}
\usepackage{pgfplotstable}
\usetikzlibrary{matrix} 
\usepgfplotslibrary{statistics}

\setcopyright{none}
\renewcommand\footnotetextcopyrightpermission[1]{} 
\definecolor{mygreen}{rgb}{0,0.6,0}
\newcounter{theorem_counter}
\newtheorem{lemma}{Lemma}
\newcounter{lemma_counter}
\begin{document}

\title[Pipeflow]{Pipeflow: An Efficient Task-Parallel Pipeline Programming Framework using Modern C++}

\author{Cheng-Hsiang Chiu}
\affiliation{
  \department{Dept. of ECE, University of Utah}              
}
\email{u1305418@utah.edu}          

\author{Tsung-Wei Huang}
\affiliation{
  \department{Dept. of ECE, University of Utah}              
}
\email{twh760812@gmail.com}          

\author{Zizheng Guo}
\affiliation{
  \department{Dept. of CS, Peking University}              
}
\email{gzz@pku.edu.cn}          

\author{Yibo Lin}
\affiliation{
  \department{Dept. of CS, Peking University}              
}
\email{yibolin@pku.edu.cn}

\begin{abstract}
Pipeline is a fundamental parallel programming pattern.
Mainstream pipeline programming frameworks count on data abstractions
to perform pipeline scheduling.
This design is convenient for data-centric pipeline applications
but inefficient for algorithms that only exploit task parallelism in pipeline.
As a result,
we introduce a new task-parallel pipeline programming 
framework called \textit{Pipeflow}.
Pipeflow does not design yet another data abstraction but focuses on the pipeline scheduling itself,
enabling more efficient implementation of task-parallel pipeline algorithms than existing frameworks.
We have evaluated Pipeflow on both micro-benchmarks and real-world applications.
As an example,
Pipeflow outperforms oneTBB 24\% and 10\% faster in a VLSI placement and 
a timing analysis workloads that adopt pipeline parallelism to speed up runtimes,
respectively.
\end{abstract}




\maketitle

\section{Introduction}

\paragraph{Motivation:} Pipeline is a fundamental parallel pattern to model parallel
executions through a linear chain of stages.
Each stage processes a \textit{data token} after the previous stage,
applies an abstract function to that data token, 
and then resolves the dependency for the next stage.
Multiple data tokens can be processed simultaneously across different stages
whenever dependencies are met.
For example, in circuit simulation, 
some operations on a gate (e.g., NAND, OR, AND) do not depend on other gates
and thus can be done at multiple logic levels simultaneously,
while operations at the same levels require processing prior levels first~\cite{Huang_21_01}.
As modern computing applications continue to adopt pipeline parallelism in various forms,
there is always a need for new pipeline programming frameworks to streamline
the implementation complexity of pipeline algorithms.

\paragraph{Limitation of state-of-art approaches:} Recent years have seen much research on
pipeline programming frameworks to assist developers in implementing pipeline algorithms
without worrying about scheduling details.
Some famous frameworks are oneTBB~\cite{TBB}, FastFlow~\cite{FastFlow}, 
GrPPI~\cite{GrPPI}, Cilk-P~\cite{CilkP}, SPar~\cite{SPar}, and HPX-pipeline~\cite{HPX}.
While each of these frameworks has its pros and cons,
a common design philosophy is 
to achieve transparent pipeline scheduling using 
\textit{data abstractions}.
This design is convenient for data-centric pipeline applications,
but it also brings three limitations:
1) Users are forced to design their pipeline algorithms in the center of data.
As we shall give a concrete example,
many applications exhibit pipeline parallelism among \textit{tasks} rather than data.
2) Under this circumstance, users need to sort out a data mapping strategy
between applications and frameworks
to perform pipeline scheduling, although such mapping is totally redundant.
3) Existing frameworks have very limited composability with other types of parallelism, 
such as task graphs which have become
essential to many irregular parallel algorithms.

\paragraph{Key insights and contributions:} After years of research, 
we have arrived at the key conclusion that 
\textit{data abstraction and task scheduling should be decoupled from each other 
in programming pipeline parallelism}.
Consequently, 
we introduce in this paper \textit{Pipeflow}, a new task-parallel pipeline programming framework
to overcome the limitations of existing ones.
We establish Pipeflow atop the open-source task graph programming system, 
\textit{Taskflow}, developed by Huang et al.~\cite{Huang_22_01},
to leverage its powerful tasking infrastructure.
As Taskflow is being used by several important research projects~\cite{OpenRoad, Magical, DREAMPlace},
building Pipeflow on top will not only benefit existing users but also
allow us to gain timely feedback from  the Taskflow community.
We summarize our contributions as follows:

\begin{itemize}[leftmargin=*]\itemsep=2pt

\item \textbf{Programming Model.} We have introduced a new C++ programming model 
for developers to create a \textit{pipeline scheduling framework}.
Unlike existing models,
we do not provide yet another data abstraction
but a flexible framework for users to fully control their application data atop 
a task-parallel pipeline scheduling framework.

\item \textbf{Task Composition.} We have introduced a composable interface to
enable seamless integration of Pipeflow into Taskflow.
Users can combine pipeline tasks with all existing task types of Taskflow 
to express a large parallel workload in a single end-to-end task graph.

\item \textbf{Scheduling Algorithm.} We have introduced a lightweight scheduling
algorithm to schedule stage tasks across parallel lines.
Our algorithm formulates the pipeline scheduling into a task graph and thus 
can efficiently solve the scheduling problem with dynamic load balancing
using Taskflow's work-stealing runtime.

\end{itemize}

\paragraph{Experimental methodology and artifact availability:}
We have evaluated Pipeflow on both micro-benchmarks and real-world applications.
As an example, 
Pipeflow outperforms oneTBB 24\% and 10\% faster in a
VLSI placement and a timing analysis workloads that adopt
pipeline parallelism to speed up runtimes, respectively.
Pipeflow is open-source and available in Taskflow as an algorithm module~\cite{Taskflow}.

\paragraph{Limitations of the proposed approach:}
Like all programming frameworks,
Pipeflow is not perfect.
Specifically, our design choice sacrifices the expressiveness 
for data-parallel pipeline applications, 
as users need to explicitly manage data storage using Pipeflow's runtime methods.
Yet, our experience leads us to believe that 
this challenge can be mitigated by deriving an application-dependent data abstraction
from our pipeline programming interface.

\section{Background}
\label{sec::background}

We first review mainstream pipeline models and 
detail the motivation of Pipeflow.
We then argue that a new task-parallel pipeline programming model is needed for
many important industrial and research areas, e.g., circuit design.

\subsection{Pipeline Basics}

Pipeline parallelism is commonly used to parallelize various applications,
such as stream processing, video processing, and dataflow systems.
These applications exhibit parallelism in the form of a \textit{linear pipeline}, 
where a linear sequence of abstraction functions, 
namely \textit{stages}, $F = \langle f_1, f_2, \cdots, f_j \rangle$,  
is applied to an input sequence of data tokens, $D = \langle d_1, d_2, \cdots, d_i \rangle$.
A linear pipeline can be thought of a loop over the data tokens of $D$.
Each iteration $i$ processes an input token $d_i$ by
applying the stage functions $F$ to $d_i$ in order.
Depending on the number of \textit{parallel lines},
$L = \langle l_1, l_2, \cdots, l_k \rangle$,
to process data tokens,
parallelism arises when iterations overlap in time.
For instance, the execution of token $d_i$ at stage $f_j$ of line $l_k$, denoted
as $f^k_j$($d_i$), can overlap with $f^{k+1}_{j-1}$($d_{i+1}$).
A stage can be a \textit{parallel} type or a \textit{serial} type to specify whether
$f^k_j$($d_i$) can overlap with $f^{k+1}_j$($d_{i+1}$) or not.
Figure \ref{fig::pipeline-graph} shows the dependency diagram of a 3-stage (serial-parallel-serial) pipeline.

\begin{figure}[!h]
  \centering
 \centerline{\includegraphics[width=.9\columnwidth]{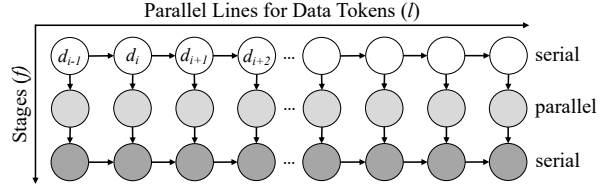}}
  \caption{
  Dependency diagram of a 3-stage (serial-parallel-serial) pipeline.
  Each node represents a task that applies a stage function to a token.
  Each edge represents a dependency between two tasks.
   }
  \label{fig::pipeline-graph}
\end{figure}

Mainstream pipeline programming libraries employ \textit{data-centric} models.
Users declare input and output data types for each stage
using library-specific data abstractions (e.g., template instantiation).
Taking oneTBB~\cite{TBB} for example, 
Listing \ref{listing::oneTBB_parallel_pipeline} implements Figure \ref{fig::pipeline-graph}
using four parallel lines and a series of callable objects called \textit{filter},
where each filter receives an input data, 
performs work on that data, and then produces a result for the next filter.
To support arbitrary application data types,
libraries typically leverage dynamic polymorphism
to allocate and convert data from a generic type (e.g., \texttt{void*}, \texttt{std::any}) to 
an application type.
To further minimize the allocation cost,
some libraries, such as oneTBB~\cite{TBB}, have implemented specialized object allocators
and buffer structures to handle temporary results between stages.

\begin{lstlisting}[language=C++,label=listing::oneTBB_parallel_pipeline,caption={oneTBB code of Figure \ref{fig::pipeline-graph}, assuming a \texttt{void}-\texttt{float}-\texttt{string}-\texttt{void} data transformation.}]
tbb::parallel_pipeline(4,  // four parallel lines
  tbb::make_filter<void,float>(
    tbb::filter_mode::serial_in_order,
    [&](tbb::flow_control& fc)-> float {
      if( data.ready() ) {
        return data.get();
      } else {
        fc.stop();
        return 0.0f;  // dummy data
      }
    }
  ) &
  tbb::make_filter<float, std::string>(
    tbb::filter_mode::parallel,
    [&](float p) { return make_string(p); }
  ) &
  tbb::make_filter<std::string,void>(
    tbb::filter_mode::serial_in_order,
    [&](std::string x) { std::cout << x; }
  )
);
\end{lstlisting}

%


\subsection{Pipeline Parallelism in CAD Algorithms}

Pipeflow is motivated by our research projects on 
developing parallel timing analysis algorithms for very large
scale integration (VLSI) computer-aided design (CAD).
Timing analysis is a critical step in the overall CAD flow
because it validates the timing performance of a digital circuit.
As design complexity continues to grow exponentially, 
the need to efficiently analyze the timing of large designs 
has become the major bottleneck to the design closure flow. 
For instance, generating a comprehensive timing report 
(e.g., pessimism removal, hundreds of corners, etc.) 
for a multi-million-gate design can take several hours to finish~\cite{ABK_18_01}.
To reduce the long analysis runtime,
recent years have seen increasing adoptions of manycore parallelism by 
new timing analysis algorithms~\cite{OpenRoad}.

\begin{figure}[!h]
  \centering
 \centerline{\includegraphics[width=1.\columnwidth]{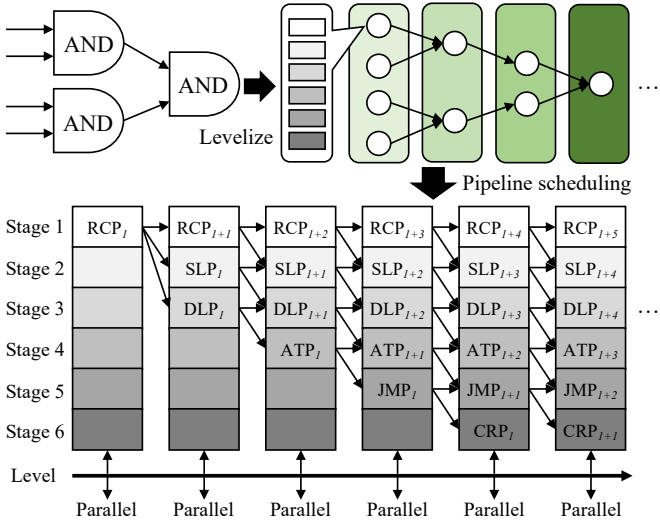}}
  \caption{
  Parallel timing propagations using pipeline~\cite{Huang_21_01}.
  Linearly dependent timing data (e.g., slew, delay, arrival time) is updated across
  graph nodes in a pipeline fashion.
   }
  \label{fig::pipeline_sta}
\end{figure}

The most widely used strategy, including commercial timers, 
to parallelize timing analysis is \textit{pipeline}.
Figure \ref{fig::pipeline_sta} illustrates
this strategy using forward timing propagation as an example~\cite{Huang_21_01}.
The circuit graph is first \textit{levelized} into a level list using topological sort.
Nodes at the same level are independent of each other and can run in parallel.
Each node runs a sequence of \textit{linearly dependent} propagation tasks,
including parasitics (RCP), slew (SLP), delay (DLP), arrival time (ATP),
jump points (JMP), and common path pessimism reduction (CRP)
to update its timing data from a custom circuit graph data structure.
Different propagation tasks can overlap across different levels using pipeline parallelism.

In fact, this type of \textit{task-parallel} pipeline strategy is ubiquitous in many
parallel CAD algorithms, such as logic simulation and physical design,
because computations frequently depend on circuit networks.
We have observed three important properties that make
mainstream pipeline programming frameworks fall short of our need:
1) Unlike the typical data-parallel pipeline, 
the pipeline parallelism in many CAD algorithms is driven by \textit{tasks} rather than data.
2) Data is not directly involved in the pipeline but the \textit{graph data structure}
crafted by a custom algorithm.
3) From user's standpoint, the real need is a \textit{pipeline scheduling framework}
to help schedule and run tasks on input tokens across parallel lines,
while leaving data management completely to applications;
in our experience, users disfavor another library data abstraction to perform pipeline scheduling,
as it often incurs development inconvenience and unnecessary data conversion overheads.


\section{Pipeflow}
\label{sec::pipeflow}

Inspired by the need of parallel CAD algorithms,
Pipeflow introduces a new task-parallel pipeline programming model
for users to create a pipeline scheduling framework
without data abstraction.
We establish Pipeflow atop the open-source parallel task graph programming system, 
Taskflow~\cite{Huang_22_01},
because it has been successfully adopted by many important CAD projects
under the DARPA IDEA program~\cite{OpenRoad, Magical, DREAMPlace}.
In this section, we will first give a brief introduction about Taskflow and then dive
into the technical details of Pipeflow.

\subsection{State of the Art: Taskflow}
\label{sec::state_of_the_art_taskflow}

Taskflow is a general-purpose parallel and heterogeneous programming system using
modern C++~\cite{Huang_22_01}.
Taskflow introduces a new \textit{control taskflow graph} (CTFG) model
that enables end-to-end implementation of task graph parallelism
coupled with in-graph control flow.
A CTFG consists of several task types, such as static task, dynamic task,
condition task, module task, runtime task, and so on.
Figure \ref{fig::conditional-tasking-do-while} shows a CTFG
of iterative tasking,
implemented in Listing \ref{listing::conditional-tasking-do-while}. 
The loop continuation condition is implemented by a single condition task, 
\texttt{cond},
that precedes two static tasks, \texttt{body} and \texttt{done}.
When \texttt{cond} returns \texttt{0}, the execution loops back to \texttt{body}.
When \texttt{cond} returns \texttt{1}, the execution moves onto \texttt{done}
and stops.
This example uses four tasks to implement a tasking loop of
100 iterations.

\begin{figure}[!h]
  \centering
  \centerline{\includegraphics[width=.9\columnwidth]{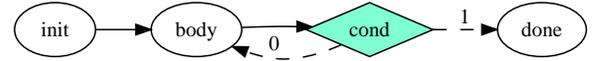}}
  \caption{A Taskflow graph of iterative control flow using one condition task and three static tasks.}
  \label{fig::conditional-tasking-do-while}
\end{figure}

\begin{lstlisting}[language=C++,label=listing::conditional-tasking-do-while,caption={Taskflow program of Figure \ref{fig::conditional-tasking-do-while}.}]
tf::Taskflow taskflow;
tf::Executor executor;
int i;
auto [init, body, cond, done] = taskflow.emplace(
  [&](){ i=0; },
  [&](){ i++; },
  [&](){ return i<100 ? 0 : 1; },
  [&](){ std::cout << "done"; }
);
init.precede(body);
body.precede(cond);
cond.precede(body, done);
executor.run(taskflow).wait();
\end{lstlisting}

Another powerful feature of Taskflow is \textit{composable tasking}.
Composable tasking enables users to define task hierarchies and compose large task
graphs from modular and reusable algorithm blocks that are easier to optimize.
Figure \ref{fig::ComposableTasking} gives an example of 
a Taskflow graph using composition.
The top-level taskflow defines one static task \texttt{C} that runs before
a dynamic task \texttt{D} that spawns two dependent tasks \texttt{D1} and \texttt{D2}.
Task \texttt{D} precedes a module task \texttt{E} composed of a taskflow of 
two dependent tasks \texttt{A} and \texttt{B}.
Listing \ref{listing::ComposableTasking} shows the 
Taskflow code of Figure \ref{fig::ComposableTasking}.
It declares two taskflows, \texttt{taskflow1} and \texttt{taskflow2}.
The second taskflow defines a \texttt{module} task that is composed of the first taskflow,
preceded by task \texttt{D}.
A module task does not own the taskflow but maintains a soft mapping to the taskflow. 
Users can create multiple module tasks from the same taskflow,
but they must not run concurrently.

\begin{figure}[!h]
  \centering
  \centerline{\includegraphics[width=.8\columnwidth]{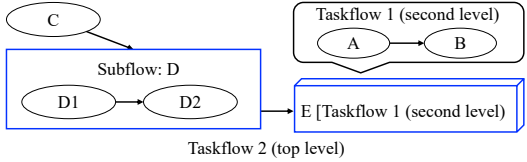}}
  \caption{An example of taskflow composition.}
  \label{fig::ComposableTasking}
\end{figure}

\begin{lstlisting}[language=C++,label=listing::ComposableTasking,caption={Taskflow code of Figure \ref{fig::ComposableTasking}.}]
tf::Taskflow taskflow1, taskflow2;
auto [A, B] = taskflow1.emplace(
  [] () { std::cout << "TaskA"; },
  [] () { std::cout << "TaskB"; }
);  
auto [C, D] = taskflow2.emplace(
  [] () { std::cout << "TaskC"; },
  [] (tf::Subflow& sf) { 
    std::cout << "TaskD"; 
    auto [D1, D2] = sf.emplace(
      [] () { std::cout << "D1"; },
      [] () { std::cout << "D2"; }
    );  
    D1.precede(D2);
  }   
);  
auto E = taskflow2.composed_of(taskflow1);
A.precede(B);
C.precede(D);
D.precede(E);
\end{lstlisting}

\paragraph{Integration with Taskflow:} 
Pipeflow leverages conditional tasking and composable tasking to 
implement the pipeline algorithm in a \textit{module task}.
Unlike existing pipeline programming frameworks that often operate 
in a standalone programming environment,
Pipeflow is designed to work seamlessly with Taskflow.
We make this architecture-level decision for three reasons:
First, the tasking patterns of parallel CAD algorithms are massive and \textit{irregular}.
Pipeline is just one part and needs to work with other tasks, 
such as graph traversal and control flow,
to compose the whole application algorithm.
Second, integrating Pipeflow into Taskflow
enables a unified scheduling runtime 
with dynamic load balancing and improved inter-operability with other Taskflow tasks.
Third,
from the ease of use standpoint,
existing Taskflow users need not to learn a different set of application programming interface
(API) but the scheduling concept of Pipeflow to implement pipeline algorithms.
While Pipeflow is primarily designed as an algorithm module of Taskflow,
we believe many of our ideas are applicable to other 
task-based parallel programming frameworks.

\subsection{Programming Model}
\label{sec::pipeflow_programming_model}

Pipeflow leverages modern C++ and template techniques to strike a balance between
expressiveness and generality.
Listing \ref{listing::Taskflow-parallel-pipeline} shows the Pipeflow counterpart of 
the oneTBB code in Listing \ref{listing::oneTBB_parallel_pipeline}
that implements the pipeline in Figure \ref{fig::pipeline-graph}.
There are three steps to create a Pipeflow application,
1) define the pipeline structure using template instantiation,
2) define the data storage, if needed, and
3) define the pipeline task using taskflow composition.
Users define the number of parallel lines and the abstract function of each stage 
in a \texttt{tf::Pipeline} object.
For each stage, users define the stage type and a pipe callable using \texttt{tf::Pipe}.
A stage can be either a serial type (\texttt{tf::PipeType::SERIAL}) or 
a parallel type (\texttt{tf::PipeType::PARALLEL}).
The pipe callable takes an argument of \texttt{tf::Pipeflow} type which is created
by the scheduler at runtime.
A \texttt{tf::Pipeflow} object represents a \textit{scheduling token} and contains several
extensible methods for users to query the runtime statistics of that token,
including the line, pipe, and token numbers.
In Pipeflow, pipe and stage are interchangeable.

\begin{lstlisting}[language=C++,label=listing::Taskflow-parallel-pipeline,caption={Pipeflow code of Figure \ref{fig::pipeline-graph}.}]
tf::Taskflow taskflow;
tf::Executor executor;
const size_t num_lines  = 4;
std::variant<float, std::string> dtype;
std::array<dtype, num_lines> buf;
tf::Pipeline pl(num_lines,
  tf::Pipe{tf::PipeType::SERIAL,
    [&](tf::Pipeflow& pf) {
      if ( !data.ready() ) {
        pf.stop();
      } else {
        buf[pf.line()] = data.get();
      }
    }
  },
  tf::Pipe{tf::PipeType::PARALLEL,
    [&](tf::Pipeflow& pf) {
      buf[pf.line()] = 
      make_string(std::get<0>(buf[pf.line()]));
    }
  },
  tf::Pipe{tf::PipeType::SERIAL,
    [&](tf::Pipeflow& pf) {
      std::cout << std::get<1>(buf[pf.line()]);
    }
  }
);
auto pipeline = taskflow.composed_of(pl);
executor.run(taskflow).wait();
\end{lstlisting}

Pipeflow does not have any data abstraction but gives applications full control over 
data management.
In our example, 
since the first and the second pipes generate \texttt{float} and \texttt{std::string} outputs, 
respectively,
we create a one-dimensional (1D) array, \texttt{buf}, to store data in a uniform storage
using \texttt{std::variant<float, std::string>}.
The dimension of the array is equal to the number of parallel lines,
as Pipeflow schedules only one token per line.
Each entry \texttt{buf[i]} stores the data that is being processed at line $i$,
which can be retrieved by \texttt{tf::Pipeflow::line}.
This organization is very space-efficient because we use only
1D array to represent data processing in a two-dimensional (2D) scheduling map.
Additionally,
by delegating data management to applications,
we can avoid dynamic data conversion between the library and the application,
which typically counts on virtual function calls
to convert a generic type (e.g., \texttt{void*}, \texttt{std::any}) to 
an arbitrary user type~\cite{TBB, FastFlow}.

Based on the above pipeline structure and data layout,
we instantiate a \texttt{tf::Pipeline} object, \texttt{pl}.
This template-based design enables the compiler to optimize
each pipe type, such as using fixed-layout
functor to store the callable and its captured data.
Finally, we create a pipeline module task \texttt{pipeline} with \texttt{pl} 
using the taskflow composition
method \texttt{composed\_of} and submit this taskflow to an executor to run the pipeline.

\begin{lstlisting}[language=C++,label=listing::Taskflow-parallel-scalable-pipeline,caption={Scalable pipeline model in Pipeflow to accept variable assignments of pipes.}]
using P =
  tf::Pipe<std::function<void(tf::Pipeflow&)>>;
std::vector<P> p(6, create_pipe());  // pipes
tf::ScalablePipeline pl(4, p.begin(), p.end());
taskflow.composed_of(pl);
executor.run(taskflow).wait();
p.resize(3);
pl.reset(p.begin(), p.end());
executor.run(taskflow).wait();
\end{lstlisting}

\texttt{tf::Pipeline} requires instantiation of all pipes at the construction time.
While this design gives compilers freedom to optimize the layout of each pipe type,
it prevents applications from varying the pipeline structure at runtime;
for instance, the number of pipes might depend on the problem size, which can be runtime variables.
To overcome this limitation,
Pipeflow provides a scalable alternative, \texttt{tf::ScalablePipeline}, to allow
variable assignments of pipes using range iterators.
In Listing \ref{listing::Taskflow-parallel-scalable-pipeline},
we create a scalable pipeline, \texttt{pl}, from a vector of six pipes.
After the first run,
we reset \texttt{pl} to another range of three pipes for the next run.
A scalable pipeline is thus more flexible 
for applications to create pipeline scheduling framework
with dynamic structures.

\subsection{Pipeline Task Composition}
\label{sec::pipeline_task_composition}

A key advantage of Pipeflow is its composability with Taskflow.
By encapsulating a pipeline in a module task,
we enable seamless integration with all existing task types in Taskflow.
This result largely facilitates the implementation of complex 
pipeline applications that require intensive interaction
with different types of task parallelism. 
Figure \ref{fig::pipeline-condition} shows a Taskflow graph that
emulates a data streaming application using a pipeline module task
and a condition task.
The condition task is used to decide if the pipeline needs to be run again
depending on the application control flow. 
Listing \ref{listing::pipeline-condition} shows the Taskflow code
of Figure \ref{fig::pipeline-condition}, using the pipeline task
in Listing \ref{listing::Taskflow-parallel-pipeline}.
When the condition task \texttt{cond} returns \texttt{0},
it informs the scheduler to rerun the pipeline task \texttt{pl},
or proceeds to \texttt{done} to stop the program otherwise.


\begin{figure}[!h]
  \centering
  \centerline{\includegraphics[width=.8\columnwidth]{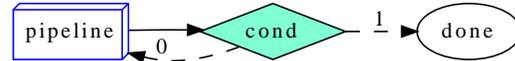}}
  \caption{A Taskflow graph of an iterative streaming application using one pipeline module task,
  one condition task, and one static task.}
  \label{fig::pipeline-condition}
\end{figure}

\begin{lstlisting}[language=C++,label=listing::pipeline-condition,caption={Taskflow code of Figure \ref{fig::pipeline-condition} using the pipeline of Listing \ref{listing::Taskflow-parallel-pipeline}.}]
auto cond = taskflow.emplace([&](){
  if ( data.ready() ) {
    std::cout << "rerun the pipeline"; 
    return 0;
  } else {
    return 1;
  }
});
auto done = taskflow.emplace([&](){ 
  std::cout << "stop"; 
});
pipeline.precede(cond);
cond.precede(pipeline, done);
\end{lstlisting}

Figure \ref{fig::taskflow-processing-pipeline} demonstrates another common application
that embeds task graph parallelism inside a pipeline.
The pipeline consists of three serial stages and four parallel lines.
Each scheduled token runs a taskflow that implements a stage algorithm in the pipeline.
The three taskflows are self-explanatory.
Different taskflows can overlap across different lines, but
only one taskflow runs on the same stage due to the serial type.
Listing \ref{listing::taskflow-processing-pipeline} 
implements Figure \ref{fig::taskflow-processing-pipeline}. 
We create a 1D array, \texttt{buf}, to store the three taskflows (defined elsewhere). 
In each stage, we obtain its taskflow at \texttt{buf[pf.pipe()]},
submit it to the executor,
and wait until the execution finishes.
As we need only four tokens,
the first pipe stops the scheduler at the fifth.

\begin{figure}[!h]
  \centering
  \centerline{\includegraphics[width=.9\columnwidth]{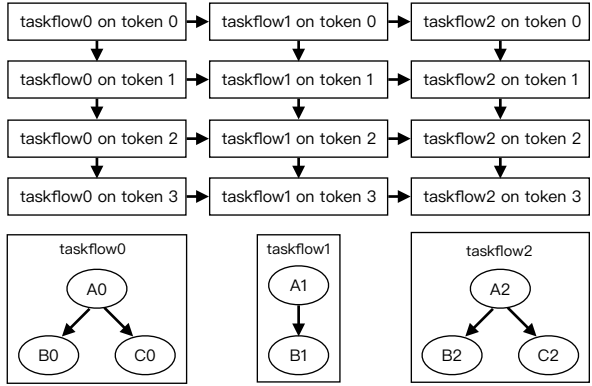}}
  \caption{A pipeline of embedded taskflows.
  Each taskflow implements a parallel algorithm at a stage in the pipeline.}
  \label{fig::taskflow-processing-pipeline}
\end{figure}

\begin{lstlisting}[language=C++,label=listing::taskflow-processing-pipeline,caption={Pipeflow and Taskflow code of Figure \ref{fig::taskflow-processing-pipeline}.}]
tf::Taskflow taskflow;
tf::Executor executor;
const size_t num_lines = 4;
const size_t num_pipes = 3;
std::array<tf::Taskflow, num_pipes> buf;
tf::Pipeline pl(num_lines,
  tf::Pipe{tf::PipeType::SERIAL, 
    [&](tf::Pipeflow& pf) {
      if (pf.token() == 4) {
        pf.stop();
        return;
      }
      executor.run(buf[pf.pipe()]).wait();
    }
  },
  tf::Pipe{tf::PipeType::SERIAL,
    [&](tf::Pipeflow& pf) {
      executor.run(buf[pf.pipe()]).wait();
    }
  },
  tf::Pipe{tf::PipeType::SERIAL,
    [&](tf::Pipeflow& pf) {
      executor.run(buf[pf.pipe()]).wait();
    }
  }
);
auto init = taskflow.emplace([](){ 
  std::cout << "init"; 
});
auto pipeline = taskflow.composed_of(pl);
auto done = taskflow.emplace([](){
  std::cout << "stop"; 
});
init.precede(pipeline);
pipeline.precede(done);
executor.run(taskflow).wait();
\end{lstlisting}

\subsection{Scheduling Algorithm}
\label{sec::scheduling_algorithm}

Pipeflow leverages Taskflow's work-stealing runtime to design an efficient
scheduling algorithm with dynamic load balancing.
As Pipeflow does not touch data abstraction,
we can simplify the pipeline scheduling problem
to deciding which task to run at which pipe and line.
Similar to oneTBB, the key idea of our scheduling algorithm is to enable only one scheduling token 
per line and process all tokens in a circular fashion across all parallel lines.
Based on the idea,
we formulate the pipeline scheduling into a lightweight Taskflow graph 
where 1) one task deals with a scheduling token per line
and 2) each task decides which task to run on its next line and pipe
using simple atomic operations.


Pipeflow creates a taskflow for each pipeline module task 
using one condition task and multiple runtime tasks one per line.
A runtime task is a task type in Taskflow for users
to interact with the executor~\cite{Huang_22_01},
such as scheduling a task in the graph.
The condition task decides which runtime task to run when the pipeline starts.
A runtime task deals with a scheduling token at a line and will 
create a pipeflow object (of type \texttt{tf::Pipeflow}) to pass 
to the pipe callable.
Figure \ref{fig::pipeflow-task-graph} shows the taskflow
of the pipeline module task in Listing \ref{listing::taskflow-processing-pipeline}.
As there are four parallel lines,
the task graph consists of one condition task, \texttt{cond}, and four runtime tasks,
\texttt{rt-0}, \texttt{rt-1}, \texttt{rt-2}, and \texttt{rt-3}.
Ultimately, only five tasks are used to run the pipeline, even though
the execution can involve many scheduling tokens.

 
\begin{figure}[!h]
  \centering
  \centerline{\includegraphics[width=.6\columnwidth]{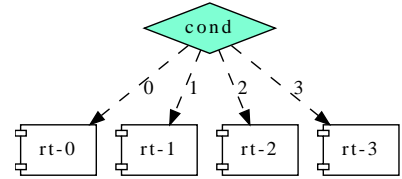}}
  \caption{The Taskflow graph of the pipeline module task in Listing \ref{listing::taskflow-processing-pipeline}.}
  \label{fig::pipeflow-task-graph}
\end{figure}

Algorithm \ref{alg::build_task_graph} implements the construction of the Taskflow graph
for a pipeline.
First, we define the condition task to return on the index of the next task
(line \ref{alg::build_task_graph::condition-task}).
Since the pipeline is running in a circular fashion,
the index is equal to the remainder of the total number of scheduled tokens
divided by the number of parallel lines.
Next, we define the runtime task using \texttt{build\_runtime\_task} for each line
(line \ref{alg::build_task_graph::runtime-task-l})
and specify the dependency between each runtime task and the condition task
(line \ref{alg::build_task_graph::dependency}).
 
\begin{algorithm}[!h]
 \SetKwInput{KwNumTokensGlobal}{global}
 \KwNumTokensGlobal{$num\_tokens$: the number of tokens}
 \SetKwInput{KwNumLinesGlobal}{global}
 \KwNumLinesGlobal{$num\_lines$: the number of lines}
 \SetKwInput{KwTasksGlobal}{global}
 \KwTasksGlobal{$tasks$: a vector of tasks}
 \SetKw{KwGoTo}{goto}
 \SetKw{KwAnd}{and}
 \SetKw{KwTrue}{true}
 \SetKw{KwFalse}{false}
 \SetKw{KwNIL}{NIL}

 $tasks[0] \leftarrow create\_condition\_task($\   \label{alg::build_task_graph::condition-task}
 $[]()\{ \Return \: num\_tokens \% num\_lines$\})\;
 \ForEach{line $l \in num\_lines$}{   \label{alg::build_task_graph::runtime-tasks}
   $build\_runtime\_task(l)$ \;       \label{alg::build_task_graph::runtime-task-l}
   $tasks[0]$.precede($tasks[l+1]$);  \label{alg::build_task_graph::dependency}
 }
 \caption{build\_task\_graph()}
 \label{alg::build_task_graph}
\end{algorithm}

\begin{algorithm}[!h]
 \SetKwInput{KwPipeflowObjectsGlobal}{global}
 \KwPipeflowObjectsGlobal{$pipeflows$: a vector of Pipeflow objects}
 \SetKwInput{KwLinesGlobal}{global}
 \KwLinesGlobal{$jcs$: a 2D array of join counters}
 \SetKwInput{KwNumTokensGlobal}{global}
 \KwNumTokensGlobal{$num\_tokens$: the number of tokens}
 \SetKwInput{KwNumLinesGlobal}{global}
 \KwNumLinesGlobal{$num\_lines$: the number of lines}
 \SetKwInput{KwNumPipesGlobal}{global}
 \KwNumPipesGlobal{$num\_pipes$: the number of pipes}
 \SetKwInput{KwTasksGlobal}{global}
 \KwTasksGlobal{$tasks$: a vector of tasks in Algorithm \ref{alg::build_task_graph}}
 \KwIn{$l$: an integer}
 \SetKw{KwGoTo}{goto}
 \SetKw{KwAnd}{and}
 \SetKw{KwTrue}{true}
 \SetKw{KwFalse}{false}
 \SetKw{KwNIL}{NIL}

 $pf \leftarrow pipeflows[l]$\;  \label{alg::build-runtime-task::pf}
 $AtomStore($\  \label{alg::build-runtime-task::goto}
 $lines[pf.line][pf.pipe].jc, jc\_of\_pf.type)$\;  
 \If{$pf.pipe == 0$} { \label{alg::build-runtime-task::first-stage}
   $pf.token \leftarrow num\_tokens$\; \label{alg::build-runtime-task::update-token}
   invoke\_pipe\_callable($pf.pipe$, $pf$)\; \label{alg::build-runtime-task::invoke-pipe-1}
   \If{$pf.stop ==$ True} {
     \Return;
   }                           \label{alg::build-runtime-task::pipe-stop-true}
   $Increment(num\_tokens)$\; \label{alg::build-runtime-task::add-token}
 }
 \If{$pf.pipe$ != $0$} {  \label{alg::build-runtime-task::not-first-stage}
   invoke\_pipe\_callable($pf.pipe$, $pf$)\;  \label{alg::build-runtime-task::invoke-pipe-2}
 }                            \label{alg::build-runtime-task::not-first-stage-end}

 $curr\_pipe \leftarrow pf.pipe$\;  \label{alg::build-runtime-task::jc-update}
 $next\_pipe \leftarrow (pf.pipe+1)\% num\_pipes$\;  
 $next\_line \leftarrow (pf.line+1)\% num\_lines$\;  
 $pf.pipe \leftarrow next\_pipe$\;
 $next\_tasks = \{\}$\;
 \If{$curr\_pipe$ is SERIAL \KwAnd 
 $AtomDec(lines[next\_line][curr\_pipe].jc) == 0$}{ \label{alg::build-runtime-task::atom-dec-downward}
   $next\_tasks$.insert($1$);    \label{alg::build-runtime-task::downward-dependency}
 }
 \If{$AtomDec(lines[pf.line][next\_pipe]) == 0$} {  \label{alg::build-runtime-task::atom-dec-right}
   $next\_tasks$.insert($0$);    \label{alg::build-runtime-task::right-dependency}
 }                             \label{alg::build-runtime-task::jc-update-end}

 \If{next\_task.$size == 2$}{  \label{alg::build-runtime-task::two-next-tasks}
   call\_scheduler($tasks[next\_line+1]$)\;  \label{alg::build-runtime-task::another-worker}
   \KwGoTo Line \ref{alg::build-runtime-task::goto}; \label{alg::build-runtime-task::current-worker} 
 }                              \label{alg::build-runtime-task::two-next-tasks-end}

 \If{next\_task.$size == 1$}{   \label{alg::build-runtime-task::one-next-task}
  \If{next\_task$[0] == 1$}{
    $pf \leftarrow pipeflows[next\_line]$; \label{alg::build-runtime-task::move-to-next-line}
  }
  \KwGoTo Line \ref{alg::build-runtime-task::goto};  
 }                              \label{alg::build-runtime-task::one-next-task-end}
 \caption{build\_runtime\_task($l$)}
 \label{alg::build-runtime-task}
\end{algorithm}

Algorithm \ref{alg::build-runtime-task} implements \texttt{build\_runtime\_task}.
When a runtime task is scheduled, we need to know which stage at which line 
for the scheduling token to work.
We keep the line and stage information in a Pipeflow object.
Each runtime task owns a Pipeflow object \texttt{pf} of a specific line
(line \ref{alg::build-runtime-task::pf}). 
Once a scheduling token is done, there are two cases for its runtime task to proceed:
1) for a parallel type, the runtime task moves to the next stage at the same line;
2) for a serial type, the runtime task additionally checks if it can move to the next line.
To carry out such a dependency constraint,
each stage keeps a join counter of an \textit{atomic} integer to represent
its dependency value.
The values of a serial stage and a parallel stage can be up to 2 and 1, respectively.
We create a 2D array \texttt{jcs} to store the join counter
of each stage at each line.
Line \ref{alg::build-runtime-task::goto} initializes these join counters to either 2 or 1
based on the corresponding stage types
that are enumerated on integer constants, 2 (serial) and 1 (parallel).
At the first stage (line \ref{alg::build-runtime-task::first-stage}),
the Pipeflow object updates its token number (line \ref{alg::build-runtime-task::update-token})
and checks if the pipe callable requests to stop the pipeline
(lines \ref{alg::build-runtime-task::invoke-pipe-1}:\ref{alg::build-runtime-task::pipe-stop-true}).
If continued, we increment the number of scheduled tokens by one
(line \ref{alg::build-runtime-task::add-token}).
For other stages, we simply invoke the pipe callables
(lines \ref{alg::build-runtime-task::not-first-stage}:\ref{alg::build-runtime-task::not-first-stage-end}).


After the pipe callable returns, we update the join counters based on the stage type
and determine the next possible tasks to run
(lines \ref{alg::build-runtime-task::jc-update}:\ref{alg::build-runtime-task::jc-update-end}).
When the join counter of a stage becomes 0,
we bookmark this stage as a task to run next
(line \ref{alg::build-runtime-task::downward-dependency}
and line \ref{alg::build-runtime-task::right-dependency}).  
If two tasks exist (line \ref{alg::build-runtime-task::two-next-tasks}),
the current runtime task informs the scheduler to call a worker thread from the executor
to run the task at the next line (line \ref{alg::build-runtime-task::another-worker})
and reiterates itself on the next pipe (line \ref{alg::build-runtime-task::current-worker}).
The idea here is to facilitate data locality as applications tend to
deal with the next stage as soon as possible.
If there is only one task available,
the current runtime task directly runs the next task with the updated \texttt{pf} object
(lines \ref{alg::build-runtime-task::one-next-task}:\ref{alg::build-runtime-task::one-next-task-end}).

\begin{figure}[!h]
  \centering
  \centerline{\includegraphics[width=.9\columnwidth]{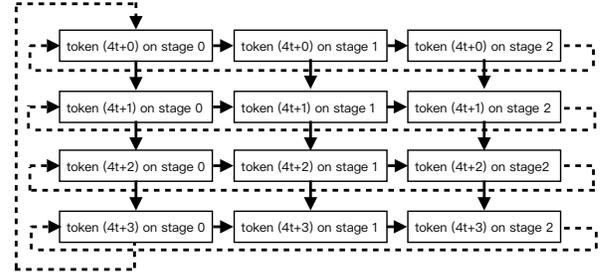}}
  \caption{The scheduling diagram of the pipeline in Listing \ref{listing::taskflow-processing-pipeline}.
  Each line runs one scheduling token. Multiple lines overlap tokens in a circular fashion.}
  \label{fig::pipeflow-scheduling-structure}
\end{figure}

Figure \ref{fig::pipeflow-scheduling-structure} illustrates our scheduling algorithm
using the pipeline
in Listing \ref{listing::taskflow-processing-pipeline}.
Since the pipeline runs in a circular fashion,
there are four dependencies (dashed edges) from the last stages to the first stages,
and one dashed edge from the first stage of the last line to the first stage of the first line. 
Each line runs only one scheduling token. Multiple lines overlap tokens in a circular fashion.
Compared to existing algorithms, such as oneTBB~\cite{TBB}, that
count on non-trivial synchronization between tasks and internal data buffers,
our algorithm focuses on the task parallelism itself and thus
largely reduces the scheduling complexity of pipeline using simple atomic operations.
We draw the following lemmas and sketch their proofs to justify 
the correctness of our scheduling algorithm:

\begin{lemma}
\label{lemma::one-runtime-one-callable}
Only one runtime task runs a pipe callable (
line \ref{alg::build-runtime-task::invoke-pipe-1} and line \ref{alg::build-runtime-task::invoke-pipe-2}
in Algorithm \ref{alg::build-runtime-task}) on a scheduling token.
\end{lemma}

\begin{proof}
Assuming there are two runtime tasks
running the same pipe callable, this means one runtime task reiterates its execution from
the previous stage and the other runtime task comes from the previous line.
In a parallel stage, this is not possible as there is no dependency from the previous line;
only one runtime task decrements the join counter to 0
(line \ref{alg::build-runtime-task::atom-dec-right} in Algorithm \ref{alg::build-runtime-task}).
In a serial stage, this is also not possible because the dependency
is resolved using atomic operations; only one runtime task will acquire the zero value
of the join counter  
(line \ref{alg::build-runtime-task::atom-dec-downward} in Algorithm \ref{alg::build-runtime-task}).
\end{proof}

\begin{lemma}
\label{lemma::every-stage}
The scheduler does not miss any stage.
\end{lemma}

\begin{proof}
We consider the situation where one runtime task moves to the next line
(line \ref{alg::build-runtime-task::move-to-next-line} in Algorithm \ref{alg::build-runtime-task})
instead of the next stage at the same line. 
Under this circumstance,
we need to make sure one runtime task will run that next stage.
Take Figure \ref{fig::pipeflow-scheduling-structure} for example,
suppose a runtime task finishes 
token \texttt{4t+1} at stage \texttt{1} and precedes to token \texttt{4t+2} on stage \texttt{1},
meaning that the join counter of token \texttt{4t+1} at stage \texttt{2} is not 0 yet.
Another runtime task that works on token \texttt{4t+0} at stage \texttt{2} will eventually
decrement the join counter to run it
(line \ref{alg::build-runtime-task::current-worker} in Algorithm \ref{alg::build-runtime-task})
or invoke another worker thread to
run it (line \ref{alg::build-runtime-task::another-worker} in Algorithm \ref{alg::build-runtime-task}).
\end{proof}

%

\section{Experimental Results}
\label{sec::experimental_results}

We evaluate the performance of Pipeflow on two fronts:
micro-benchmarks and two real-world industrial CAD applications.
We study the performance across runtime, memory (RSS), and throughput.
We do not use conventional pipeline benchmarks as their sizes are relatively small
compared to CAD (e.g., 6 pipes in ferret\cite{PaRSEC}).
The runtime difference between Pipeflow and the baseline is not obvious on small pipelines.
We compiled all programs using clang++ v10 with C++17 standard 
\texttt{-std=c++17} and optimization flag \texttt{-O3} enabled. 
We run all the experiments on a Ubuntu Linux 5.3.0-64-generic x86 machine
with 40 Intel Xeon CPU cores at 2.00 GHz and 256 GB RAM.
Each application thread corresponds to one CPU core.
All data is an average of five runs.

\subsection{Baseline}
\label{sec::baseline}

Given a large number of pipeline programming frameworks, it is infeasible to compare
Pipeflow with all of them. 
Each of them has its pros and cons and dominates certain applications.
Since Pipeflow is inspired by our CAD applications,
we select oneTBB Parallel Pipeline (v2021.5.0)~\cite{TBB} as our baseline,
which has been widely used in the CAD community.
We believe this selection is sufficient and fair to highlight the advantage of Pipeflow,
considering the similar tasking infrastructure between oneTBB and Pipeflow.
For oneTBB pipelines, we pass a nominal integer between successive stages (i.e., \texttt{tbb::filter}),
because oneTBB does not allow void type but
implements a specialized memory allocator to hold intermediate results returned by stages.

\subsection{Micro-benchmarks}
\label{sec::micro_benchmarks}

The purpose of micro-benchmarks is to measure the pure scheduling performance
of Pipeflow without much computation bias from the application.
We compare the runtime and memory between Pipeflow and oneTBB 
for completing pipelines of different numbers of serial stages, scheduling tokens, and threads.
We do not use parallel stages as the callable of a parallel pipe can be absorbed into the
previous serial pipe.
Each stage performs a nominal work of constant space and time complexity and 
forwards the scheduling token to the next stage.

\input{Fig/micro_benchmark/runtime-memory-data-scalability.tex}

Figure \ref{fig::micro-benchmark-data-scalability} draws the runtime improvement
of Pipeflow over oneTBB and memory data
under different numbers of scheduling tokens.
Here, we use 16 threads, which produce the best performance for both,
to run a pipeline of 80 serial stages and 80 parallel lines.
With more token numbers ($>$256),
Pipeflow is consistently better than oneTBB, despite the slight improvement
($<$1\%).
When the number of scheduling tokens is small, the variation is large
(e.g., up to 6\% improvement at 32 tokens).
This is because when oneTBB starts the pipeline, it requires expensive
set-up time on the data buffers, whereas Pipeflow can immediately
start the task scheduling.
Yet, as we increase the number of tokens, such cost can be amortized.
In terms of memory, oneTBB is always higher than Pipeflow
(e.g., 21\% at 65K tokens) 
since we do not manage any data buffers but focus on the task scheduling itself.

\input{Fig/micro_benchmark/runtime-memory-pipe-scalability.tex}

Figure \ref{fig::micro-benchmark-pipe-scalability} draws the runtime
and memory performance using 16 threads to schedule 65K tokens 
through different numbers of stages, 
where the number of parallel lines is equal to the stage count.
We observe that Pipeflow is consistently faster than oneTBB when
the number of stages is larger than 16, despite the difference being slight ($<$1\%).
At 8 stages, the available task parallelism is smaller than the given
thread parallelism, and oneTBB is faster in this scenario.
This is due to Taskflow's scheduling algorithm. Taskflow always keeps one thread
busy in stealing while there is an active worker.
When the available task parallelism is scarce,
the scheduling cost becomes expensive. 
Yet, this problem can be mitigated by users selecting the right line number
in their pipeline applications.
For instance, beyond 8 stages, Pipeflow starts to outperform oneTBB.
In terms of memory,
we can clearly see the difference between Pipeflow and oneTBB
(e.g., 19\% at 80 stages).
As we increase the number of stages,
oneTBB needs more space for internal data buffers, whereas
Pipeflow delegates the data management completely to applications.

\input{Fig/micro_benchmark/runtime-memory-thread-scalability.tex}

Figure \ref{fig::micro-benchmark-thread-scalability} shows the runtime
and memory results at different numbers of threads to run 65K tokens 
on a pipeline of 80 stages and 80 parallel lines.
Both Pipeflow and oneTBB scale equally well as the number of cores increases.
At each point, we observe a small win of Pipeflow.
For instance, at 32 cores, Pipeflow is 4.3\% faster than oneTBB.
In terms of memory,
both Pipeflow and oneTBB use more memory with more threads.
However, there remains a clear gap between Pipeflow and oneTBB
(e.g., 22\% less by Pipeflow at 40 cores).

\input{Fig/micro_benchmark/throughput-work.tex}

Figure \ref{fig::micro-benchmark-throughput} compares the throughput by corunning the
same program up to 10 times. 
We use the \textit{weighted speedup} to measure the system throughput, which is the sum
of the individual speedup of each process over a baseline execution time~\cite{BWS}.
A throughput of one implies that the corun throughput is the same 
as if those processes run consecutively.
On the left, the pipeline has 8 stages and 8 parallel lines.
On the right, the pipeline has 80 stages and 80 parallel lines.
Both of them run 65K scheduling tokens using 40 threads.
The experiment emulates a server-like environment where different pipeline applications
compete for the same resources.
We can see that Pipeflow outperforms oneTBB in most coruns.
For example, at 10 coruns, Pipeflow is 1.1 $\times$ 
and 1.5$\times$ better than oneTBB 
with 8 and 80 stages, respectively.

\input{Fig/graph_pipeline/runtime-memory-pipe-scalability-80L.tex}

\subsection{VLSI Circuit Timing Analysis Algorithm}
\label{sec::graph_pipeline}

We applied Pipeflow to solve a VLSI static timing analysis (STA) problem. 
The goal is to analyze the timing landscape of a circuit design
and report critical paths that do not meet the given constraints (e.g., setup and hold).
As presented in Figure~\ref{fig::pipeline_sta},
modern STA engines leverage pipeline parallelism to speed up the timing propagations.
However, nearly all of them count on OpenMP-based loop parallelism 
with layer-by-layer synchronization~\cite{Huang_21_01}.
With Pipeflow, we can directly formulate the problem as a task-parallel pipeline 
to improve task asynchrony.
As the analysis complexity continues to increase,
more analysis tasks (e.g., RC, delay calculators, pessimism reduction)
are incorporated into each node in the STA graph.
These tasks can be encapsulated in a sequence of stage functions
to overlap in graph across parallel lines.
We modify a large circuit design of 1.5M nodes and 3.5M edges from~\cite{Huang_21_01} 
and study the performance under different stage counts.
Each node has a stage task to calculate delay values at a specific configuration
using 2D matrix multiplication.

Figure \ref{fig::graph-pipeline-pipe-scalability-80L} compares the runtime and
memory between Pipeflow and oneTBB up to 80 stages.
We use 80 parallel lines for all experiments; 
we do not observe much difference in other numbers of parallel lines as both Taskflow and oneTBB 
have an adaptive work-stealing strategy to balance the number of running threads
with dynamic task parallelism.
In general, Pipeflow outperforms oneTBB at large stage numbers.
For example, using 40 cores and 72 stages, Pipeflow is 10\% faster than oneTBB.
In terms of memory usage,
Pipeflow is always better than oneTBB regardless of the number of stages and cores.
Pipeflow consumes less memory than oneTBB does because
all stage tasks perform computations directly on a global graph data structure
captured in the pipe callable.
The data passing interface between successive stages in oneTBB thus 
becomes a significant yet unnecessary overhead.

\input{Fig/graph_pipeline/throughput.tex}

Next, we compare the throughput by corunning the same program up to 10 times.
Corunning an STA program is very common for reporting the timing data
of a design at different input library files~\cite{ABK_18_01}.
The effect of pipeline scheduling propagates
to all simultaneous processes.
Hence, throughput is a good measurement for the inter-operability of 
a pipeline-based STA algorithm.
We corun the same analysis program up to 10 processes that compete for 40 cores.
Again, we use the weighted speedup to measure the throughput.
Figure \ref{fig::graph-pipeline-throughput} plots the throughput across 10
coruns at 8 and 80 stages.
We can see that Pipeflow outperforms oneTBB at all coruns.
For instance, at 10 coruns, Pipeflow
is 1.7$\times$ and 1.2 $\times$ better than oneTBB with 8 and 80 stages, respectively.

\subsection{VLSI Detailed Placement Algorithm}

We applied Pipeflow to solve a VLSI detailed placement problem.
Detailed placement is a critical step in physical optimization.
The goal is to optimize the interconnect among millions of logic gates or \textit{instances}
for improved timing and power.
Connected instances are grouped to a \textit{net} with interconnect modeled in Manhattan distance.
We consider the detailed placement algorithm in DREAMPlace~\cite{DREAMPlace}, 
namely \textit{local reordering}.
The algorithm decides an optimal order of four consecutive instances
in a \textit{window} of a placement row that produces minimum interconnect wirelength.
We can parallelize the reordering algorithm using pipeline.
Each row corresponds to a parallel stage that finds the best ordering of cells in a window
from the top to the bottom.
The scheduling tokens sweep through all windows from the left to the right.
Figure~\ref{fig::pipeline_detailed_placement} illustrates the algorithm.

\begin{figure}[!h]
  \centering
  \centerline{\includegraphics[width=.8\columnwidth]{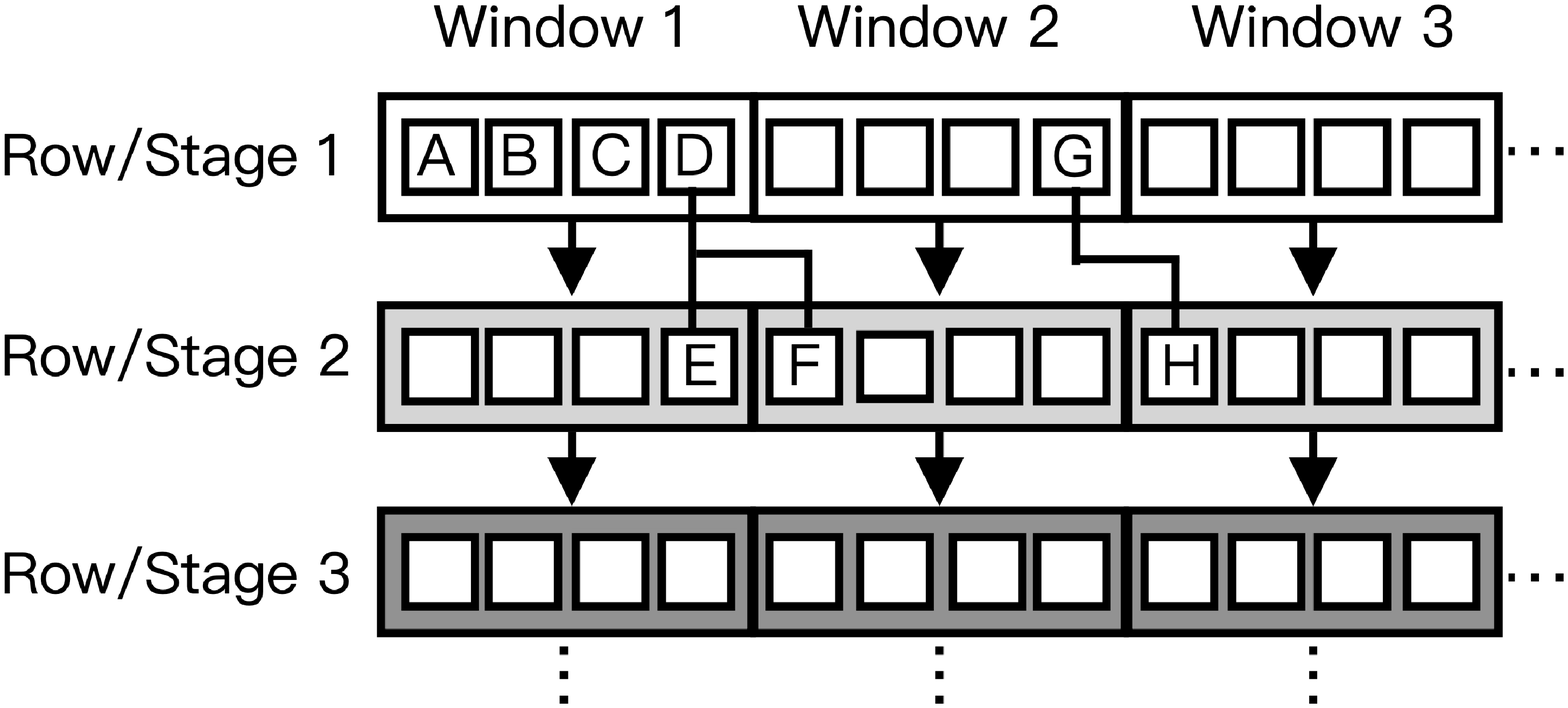}}
  \caption{
    Parallel local reordering algorithm using pipeline.
    Reordering cells in window 1 of row 1 (R1W1) cannot run in parallel
    with R2W1 due to the vertical dependency D--E.
    However, R1W1 can overlap with R2W2 despite D--F, 
    because the algorithm deals with Manhattan distance for wirelength.
    We can always assume that F is fixed on the right of window 1 within the view of R1W1.
    Thus, we update vertical windows in a linear pipeline, R1$\to$R2$\to$R3.
  }
  \label{fig::pipeline_detailed_placement}
\end{figure}

\input{Fig/dreamplace/scalability.tex}

Figure \ref{fig::dreamplace_scalability} compares the runtime and memory data 
between Pipeflow and oneTBB for two industrial designs, adaptec1 and bigblue4.
adaptec1 is a medium design with 211K instances and 221K nets (instance dependencies), 
defining 890 stages for 890 placement rows and 10692 windows.
bigblue4 is a large design with 2.1M instances and 2.2M nets, defining 2694 stages for 2694 placement rows and 32190 columns.
Both Pipeflow and oneTBB scale with increasing numbers of cores,
whereas Pipeflow always outperforms oneTBB 3--24\%.
For instance, in bigblue4 (40 cores), Pipeflow finishes the placement algorithm in 6131 seconds,
whereas oneTBB needs 8213 seconds.
Such improvement is significant as practical design closure process can
invoke millions of placement iterations to optimize the physical layout.
Likewise, Pipeflow outperforms oneTBB in memory usage.
The difference is slight (about 1\%) because most memory is taken
by the placement problem itself, including the data structure of rows and instances.

\subsection{Insight from the CAD Algorithm Developers}

As an experienced parallel CAD researcher,
Pipeflow has assisted us to overcome many programming challenges.
In the previous two experiments,
the data is explicitly managed by the application algorithms and 
does not go through any data abstraction layers of oneTBB.
The real need is a task-parallel pipeline programming framework that 1) gives us
full control over data and 2) 
allows us to probe each scheduled task.
For instance, when implementing the placement algorithm,
we capture the row data from a global database in each pipe callable and 
use the pipeflow variable to get the line numbers of a scheduled task to index 
its window locations.
However, oneTBB has abstracted these components out, and we have to implement
another mapping strategy to get these data from each filter.
Similar problems exist in other libraries too.

\section{Related Work}
\label{sec::related_work}

\paragraph{Pipeline programming models} have received intensive research interest.
Most of them are data-centric
using static template instantiation or dynamic runtime polymorphism
to model data processing in a pipeline.
To name a few popular examples: 
oneTBB~\cite{TBB} and TPL~\cite{TPL} require explicit definitions of input and output types
for each stage;
GrPPI~\cite{GrPPI} provides a composable abstraction for data- and stream-parallel 
patterns with a pluggable back-end to support task scheduling;
FastFlow~\cite{FastFlow} models parallel dataflow using pre-defined 
sequential and parallel building blocks;
TTG~\cite{TTG} focuses on dataflow programming using various template optimization techniques;
SPar~\cite{SPar, SPar2, SPar3, SPar4} analyzes annotated attributes
extracted from the data and stream parallelism domain,
and automatically generates parallel patterns defined in FastFlow;
Proteas~\cite{Proteas} introduces a programming model for directive-based parallelization
of linear pipeline;
~\cite{SPar-adaptive, Pattern-adaption} propose self-adaptive mechanism to decide the degree of parallelism
and generate the pattern compositions in FastFlow.
These programming models, however, constrain users to design pipeline algorithms
using their data models,
making it difficult to use especially 
for applications that only need pipeline scheduling atop custom data structures.

\paragraph{Existing pipeline scheduling algorithms} typically co-design
task scheduling and buffer structures to strive for the best performance.
For instance, oneTBB~\cite{TBB} defines a per-stage buffer counter to synchronize
data tokens among stages and lines, coupled with a small object allocator
to minimize the data allocation overhead;
GRAMPS~\cite{GRAMPS} designs a buffer manager with per-thread fix-sized memory pools to 
dynamically allocate new data and release used ones;
FastFlow~\cite{FastFlow} design a lock-free queue
with a mechanism to transfer data ownership between senders and receivers,
but this method can incur imbalanced load and requires non-trivial back-pressure management;
HPX~\cite{HPX} counts on a channel data structure and standard future objects
to pass data around tasks,
but the creation of share states becomes expensive when the pipeline is large;
Cilk-P~\cite{CilkP} employs per-stage queues coupled with 
two counter types to track static and dynamic dependencies
of each node, but it targets on-the-fly pipeline parallelism
which is orthogonal to our focus;
FDP~\cite{FDP} proposes a learning-based mechanism to 
adapt scheduling to an environment, but it requires expensive
runtime profiling that may not work well for highly irregular applications like CAD.
%
%
%
%
In terms of load balancing,
most pipeline schedulers leverage work stealing, 
which has been reported with better performance than static
policies~\cite{Perf-Eval, WorkStealPipeline, GRAMPS, CilkP, Cilk++, Cilk}.
However, for some special cases, such as fine-grained load-imbalanced pipelines,
static policies perform comparably.
For example, Pipelight~\cite{Pipelight} implements a load balancing technique
based on two static scheduling algorithms,
DSWP~\cite{DSWP, DSWP2, DSWP3} and LBPP~\cite{LBPP};
Pipelite~\cite{Pipelite} and URTS~\cite{URTS}
introduce dynamic schedulers using ticket-based synchronization and directive-based model language
for linearpipelines, respectively.
While co-designing task scheduling and buffer structures has certain advantages
for data-centric pipeline (e.g., data locality),
the cost of managing data can be significant yet unnecessary, especially
for applications that only exploit task parallelism in pipeline.


\section{Conclusion}
\label{sec::Conclusion}

We have introduced Pipeflow,
an efficient C++ pipeline programming framework atop
the Taskflow system.
We have designed a new task-parallel programming model that
separates data abstraction and task scheduling.
By focusing on the pipeline tasking,
we have introduced a simple yet efficient scheduling algorithm
based on Taskflow's work-stealing runtime with dynamic load balancing.
We have evaluated the performance of Pipeflow on micro-benchmarks
and real applications.
For example, 
Pipeflow outperforms oneTBB 24\% and 10\% faster in a
VLSI placement and a timing analysis workloads that adopt
pipeline parallelism to speed up runtimes, respectively.
%
Our future work plans to apply Pipeflow to more CAD applications and
bring interdisciplinary ideas to the HPC domain.
%




\Urlmuskip=0mu plus 1mu\relax
\renewcommand*{\bibfont}{\normalsize}
\bibliography{ms.bib}

%

\end{document}

%% file: Fig/micro_benchmark/runtime-memory-data-scalability.tex
\begin{figure}[!h]
  \centering
  \pgfplotsset{
    title style={font=\LARGE},
    label style={font=\LARGE},
  }
  \begin{tikzpicture}[scale=0.49]
    \begin{axis}[
      title=Runtime Improvement,
      xmode=log,
      log basis x={2},
      ylabel=Ratio (\%),
      xlabel=Number of Scheduling Tokens,
      legend pos=north west,
    ]
    \addplot+ table[x=size,y expr=100*(\thisrow{tbb_16t} - \thisrow{tf_16t})/\thisrow{tbb_16t},col sep=space]{Fig/micro_benchmark/data-scalability-runtime.txt};
    \end{axis}
  \end{tikzpicture}
  \begin{tikzpicture}[scale=0.49]
    \begin{axis}[
      title=Memory,
      xmode=log,
      log basis x={2},
      ylabel=Maximum RSS (MB),
      xlabel=Number of Scheduling Tokens,
      legend pos=north west,
      y filter/.code={\pgfmathparse{#1/1000}\pgfmathresult}
    ]
    \addplot+ table[x=size,y=tf_16t,col sep=space]{Fig/micro_benchmark/data-scalability-memory.txt};
    \addplot+ table[x=size,y=tbb_16t,col sep=space]{Fig/micro_benchmark/data-scalability-memory.txt};
    \legend{Pipeflow, oneTBB}
    \end{axis}
  \end{tikzpicture}
  \caption{Runtime (improvement of Pipeflow over oneTBB) and memory performance at different numbers of scheduling tokens
  using 16 threads, 80 parallel lines, and 80 serial stages.}
  \label{fig::micro-benchmark-data-scalability}
\end{figure}

%% file: Fig/micro_benchmark/runtime-memory-pipe-scalability.tex
\begin{figure}[!h]
  \centering
  \pgfplotsset{
    title style={font=\LARGE},
    label style={font=\LARGE},
  }
  \begin{tikzpicture}[scale=0.49]
    \begin{axis}[
      title=Runtime Improvement,
      ylabel=Ratio (\%),
      xlabel=Number of Stages,
      legend pos=north west,
    ]
    \addplot+ table[x=pipe,y expr=100*(\thisrow{tbb_16t}-\thisrow{tf_16t})/\thisrow{tbb_16t},col sep=space]{Fig/micro_benchmark/pipe-scalability-runtime.txt};
    \end{axis}
  \end{tikzpicture}
  \begin{tikzpicture}[scale=0.49]
    \begin{axis}[
      title=Memory,
      ylabel=Maximum RSS (MB),
      xlabel=Number of Stages,
      legend pos=north west,
      y filter/.code={\pgfmathparse{#1/1000}\pgfmathresult}
    ]
    \addplot+ table[x=pipe,y=tf_16t,col sep=space]{Fig/micro_benchmark/pipe-scalability-memory.txt};
    \addplot+ table[x=pipe,y=tbb_16t,col sep=space]{Fig/micro_benchmark/pipe-scalability-memory.txt};
    \legend{Pipeflow, oneTBB}
    \end{axis}
  \end{tikzpicture}
  \caption{Runtime (improvement of Pipeflow over oneTBB) and memory performance at 
  different numbers of serial stages using 16 threads and 65K scheduling tokens.
  The number of parallel lines equals the number of the stages.}
  \label{fig::micro-benchmark-pipe-scalability}
\end{figure}

%% file: Fig/micro_benchmark/runtime-memory-thread-scalability.tex
\begin{figure}[!h]
  \centering
  \pgfplotsset{
    title style={font=\LARGE},
    label style={font=\LARGE},
  }
  \begin{tikzpicture}[scale=0.45]
    \begin{axis}[
      title=Runtime,
      axis y line*=left,
      ylabel=Runtime (s),
      xlabel=Number of Cores,
      legend style={at={(0.15,0.75)},anchor=south west},
      y filter/.code={\pgfmathparse{#1/1000}\pgfmathresult}
    ]
    \addplot table[x=thread,y=tf_80l,col sep=space]{Fig/micro_benchmark/thread-scalability-runtime.txt};
    \addplot table[x=thread,y=tbb_80l,col sep=space]{Fig/micro_benchmark/thread-scalability-runtime.txt};
    \legend{Pipeflow, oneTBB}
    \end{axis}
    \begin{axis}[
      axis y line*=right,
      ylabel=Improvement Ratio (\%),
      axis x line=none,
      ylabel style = {font=\LARGE,yshift=-9cm},
    ]
    \addplot [ybar,color=black] table[x=thread,y expr=100*(\thisrow{tbb_80l}-\thisrow{tf_80l})/\thisrow{tbb_80l},col sep=space]{Fig/micro_benchmark/thread-scalability-runtime.txt};
    \end{axis}
  \end{tikzpicture}
  \begin{tikzpicture}[scale=0.45]
    \begin{axis}[
      title=Memory,
      ylabel=Maximum RSS (MB),
      xlabel=Number of Cores,
      legend pos=north west,
      y filter/.code={\pgfmathparse{#1/1000}\pgfmathresult}
    ]
    \addplot+ table[x=thread,y=tf_80l,col sep=space]{Fig/micro_benchmark/thread-scalability-memory.txt};
    \addplot+ table[x=thread,y=tbb_80l,col sep=space]{Fig/micro_benchmark/thread-scalability-memory.txt};
    \legend{Pipeflow, oneTBB}
    \end{axis}
  \end{tikzpicture}
  \caption{Runtime and memory performance at differents number of threads running 65K
  tokens on a pipeline of 80 stages and 80 parallel lines. The bars illustrate the runtime improvement
  of Pipeline over oneTBB.}
  \label{fig::micro-benchmark-thread-scalability}
\end{figure}

%% file: Fig/micro_benchmark/throughput-work.tex
\begin{figure}[!h]
  \centering
  \pgfplotsset{
    title style={font=\LARGE},
    label style={font=\LARGE},
  }
  \begin{tikzpicture}[scale=0.49]
    \begin{axis}[
      title=Corun (8 Stages),
      ylabel=Throughput,
      xlabel=Number of Coruns,
      legend pos=north east,
    ]
    \addplot+ table[x=corun,y=tf,col sep=space]{Fig/micro_benchmark/throughput-8pipes-work.txt};
    \addplot+ table[x=corun,y=tbb,col sep=space]{Fig/micro_benchmark/throughput-8pipes-work.txt};
    \legend{Pipeflow, oneTBB}
    \end{axis}
  \end{tikzpicture}
  \begin{tikzpicture}[scale=0.49]
    \begin{axis}[
      title=Corun (80 Stages),
      ylabel=Throughput,
      xlabel=Number of Coruns,
      legend pos=north east,
    ]
    \addplot+ table[x=corun,y=tf,col sep=space]{Fig/micro_benchmark/throughput-80pipes-work.txt};
    \addplot+ table[x=corun,y=tbb,col sep=space]{Fig/micro_benchmark/throughput-80pipes-work.txt};
    \legend{Pipeflow, oneTBB}
    \end{axis}
  \end{tikzpicture}
  \caption{Throughput of corunning micro-benchmark programs with 8 and 80 stages.}
  \label{fig::micro-benchmark-throughput}
\end{figure}

%% file: Fig/graph_pipeline/runtime-memory-pipe-scalability-80L.tex
\begin{figure}[!h]
  \centering
  \pgfplotsset{
    title style={font=\LARGE},
    label style={font=\LARGE},
  }
  \begin{tikzpicture}[scale=0.49]
    \begin{axis}[
      title=Runtime (16 Cores),
      ylabel=Runtime (s),
      xlabel=Number of Stages,
      legend pos=south east,
      y filter/.code={\pgfmathparse{#1/1000}\pgfmathresult}
    ]
    \addplot+ table[x=pipe,y=tf_16t_80l,col sep=space]{Fig/graph_pipeline/pipe-scalability-runtime.txt};
    \addplot+ table[x=pipe,y=tbb_16t_80l,col sep=space]{Fig/graph_pipeline/pipe-scalability-runtime.txt};
    \legend{Pipeflow, oneTBB}
    \end{axis}
  \end{tikzpicture}
  \begin{tikzpicture}[scale=0.49]
    \begin{axis}[
      title=Memory (16 Cores),
      ylabel=Maximum RSS (MB),
      xlabel=Number of Stages,
      legend pos=north west,
      ymax=196,
    ]
    \addplot+ table[x=pipe,y expr=\thisrow{tf_16t_80l}/1000,col sep=space]{Fig/graph_pipeline/pipe-scalability-memory.txt};
    \addplot+ table[x=pipe,y expr=\thisrow{tbb_16t_80l}/1000,col sep=space]{Fig/graph_pipeline/pipe-scalability-memory.txt};
    \legend{Pipeflow, oneTBB}
    \end{axis}
  \end{tikzpicture}
  \begin{tikzpicture}[scale=0.49]
    \begin{axis}[
      title=Runtime (32 Cores),
      ylabel=Runtime (s),
      xlabel=Number of Stages,
      legend pos=south east,
      y filter/.code={\pgfmathparse{#1/1000}\pgfmathresult}
    ]
    \addplot+ table[x=pipe,y=tf_32t_80l,col sep=space]{Fig/graph_pipeline/pipe-scalability-runtime.txt};
    \addplot+ table[x=pipe,y=tbb_32t_80l,col sep=space]{Fig/graph_pipeline/pipe-scalability-runtime.txt};
    \legend{Pipeflow, oneTBB}
    \end{axis}
  \end{tikzpicture}
  \begin{tikzpicture}[scale=0.49]
    \begin{axis}[
      title=Memory (32 Cores),
      ylabel=Maximum RSS (MB),
      xlabel=Number of Stages,
      legend pos=north west,
      y filter/.code={\pgfmathparse{#1/1000}\pgfmathresult}
    ]
    \addplot+ table[x=pipe,y=tf_32t_80l,col sep=space]{Fig/graph_pipeline/pipe-scalability-memory.txt};
    \addplot+ table[x=pipe,y=tbb_32t_80l,col sep=space]{Fig/graph_pipeline/pipe-scalability-memory.txt};
    \legend{Pipeflow, oneTBB}
    \end{axis}
  \end{tikzpicture}
  \begin{tikzpicture}[scale=0.49]
    \begin{axis}[
      title=Runtime (40 Cores),
      ylabel=Runtime (s),
      xlabel=Number of Stages,
      legend pos=south east,
      y filter/.code={\pgfmathparse{#1/1000}\pgfmathresult}
    ]
    \addplot+ table[x=pipe,y=tf_40t_80l,col sep=space]{Fig/graph_pipeline/pipe-scalability-runtime.txt};
    \addplot+ table[x=pipe,y=tbb_40t_80l,col sep=space]{Fig/graph_pipeline/pipe-scalability-runtime.txt};
    \legend{Pipeflow, oneTBB}
    \end{axis}
  \end{tikzpicture}
  \begin{tikzpicture}[scale=0.49]
    \begin{axis}[
      title=Memory (40 Cores),
      ylabel=Maximum RSS (MB),
      xlabel=Number of Stages,
      legend pos=north west,
      ymax=198,
      y filter/.code={\pgfmathparse{#1/1000}\pgfmathresult}
    ]
    \addplot+ table[x=pipe,y=tf_40t_80l,col sep=space]{Fig/graph_pipeline/pipe-scalability-memory.txt};
    \addplot+ table[x=pipe,y=tbb_40t_80l,col sep=space]{Fig/graph_pipeline/pipe-scalability-memory.txt};
    \legend{Pipeflow, oneTBB}
    \end{axis}
  \end{tikzpicture}
  \caption{Runtime and memory comparisons between Pipeflow and oneTBB to complete a design of 1.5M nodes
  and 3.5M edges with different numbers of serial stages.}
  \label{fig::graph-pipeline-pipe-scalability-80L}
\end{figure}

%% file: Fig/graph_pipeline/throughput.tex
\begin{figure}[!h]
  \centering
  \pgfplotsset{
    title style={font=\LARGE},
    label style={font=\LARGE},
  }
  \begin{tikzpicture}[scale=0.49]
    \begin{axis}[
      title=Corun (8 stages),
      ylabel=Throughput,
      xlabel=Number of Coruns,
      legend pos=north east,
    ]
    \addplot+ table[x=corun,y=tf,col sep=space]{Fig/graph_pipeline/throughput-8pipes.txt};
    \addplot+ table[x=corun,y=tbb,col sep=space]{Fig/graph_pipeline/throughput-8pipes.txt};
    \legend{Pipeflow, oneTBB}
    \end{axis}
  \end{tikzpicture}
  \begin{tikzpicture}[scale=0.49]
    \begin{axis}[
      title=Corun (80 Stages),
      ylabel=Throughput,
      xlabel=Number of Coruns,
      legend pos=north east,
    ]
    \addplot+ table[x=corun,y=tf,col sep=space]{Fig/graph_pipeline/throughput.txt};
    \addplot+ table[x=corun,y=tbb,col sep=space]{Fig/graph_pipeline/throughput.txt};
    \legend{Pipeflow, oneTBB}
    \end{axis}
  \end{tikzpicture}
  \caption{Throughput of corunning the STA program.}
  \label{fig::graph-pipeline-throughput}
\end{figure}

%% file: Fig/dreamplace/scalability.tex
\begin{figure}[!h]
  \centering
  \pgfplotsset{
    title style={font=\LARGE},
    label style={font=\LARGE},
  }
  \begin{tikzpicture}[scale=0.49]
    \begin{axis}[
      title= Runtime (adaptec1),
      ylabel=Runtime (s),
      xlabel=Number of Cores,
      legend pos=north east,
      xtick={1, 8, 16, 24, 32, 40},
    ]
    \addplot+ table[x=cpus,y=tf-a1-rt,col sep=space]{Fig/dreamplace/scalability.txt};
    \addplot+ table[x=cpus,y=tbb-a1-rt,col sep=space]{Fig/dreamplace/scalability.txt};
    \legend{Pipeflow, oneTBB}
    \end{axis}
  \end{tikzpicture}
  \begin{tikzpicture}[scale=0.49]
    \begin{axis}[
      title= Memory (adaptec1),
      ylabel=Maximum RSS (MB),
      xlabel=Number of Cores,
      legend pos=north west,
      xtick={1, 8, 16, 24, 32, 40},
      y filter/.code={\pgfmathparse{#1/1000}\pgfmathresult}
    ]
    \addplot+ table[x=cpus,y=tf-a1-rss,col sep=space]{Fig/dreamplace/scalability.txt};
    \addplot+ table[x=cpus,y=tbb-a1-rss,col sep=space]{Fig/dreamplace/scalability.txt};
    \legend{Pipeflow, oneTBB}
    \end{axis}
  \end{tikzpicture}
  \begin{tikzpicture}[scale=0.49]
    \begin{axis}[
      title= Runtime (bigblue4),
      ylabel=Runtime (s),
      xlabel=Number of Cores,
      legend pos=north east,
      xtick={1, 8, 16, 24, 32, 40},
    ]
    \addplot+ table[x=cpus,y=tf-b4-rt,col sep=space]{Fig/dreamplace/scalability.txt};
    \addplot+ table[x=cpus,y=tbb-b4-rt,col sep=space]{Fig/dreamplace/scalability.txt};
    \legend{Pipeflow, oneTBB}
    \end{axis}
  \end{tikzpicture}
  \begin{tikzpicture}[scale=0.49]
    \begin{axis}[
      title= Memory (bigblue4),
      ylabel=Maximum RSS (GB),
      xlabel=Number of Cores,
      legend pos=north west,
      xtick={1, 8, 16, 24, 32, 40},
      ymax=1.64,
      y filter/.code={\pgfmathparse{#1/1000000}\pgfmathresult}
    ]
    \addplot+ table[x=cpus,y=tf-b4-rss,col sep=space]{Fig/dreamplace/scalability.txt};
    \addplot+ table[x=cpus,y=tbb-b4-rss,col sep=space]{Fig/dreamplace/scalability.txt};
    \legend{Pipeflow, oneTBB}
    \end{axis}
  \end{tikzpicture}
  \caption{Runtime and memory data of the pipeline-based placement algorithm for two industrial designs, 
  adaptec1 (890 stages) and bigblue4 (2694 stages).}
  \label{fig::dreamplace_scalability}
\end{figure}
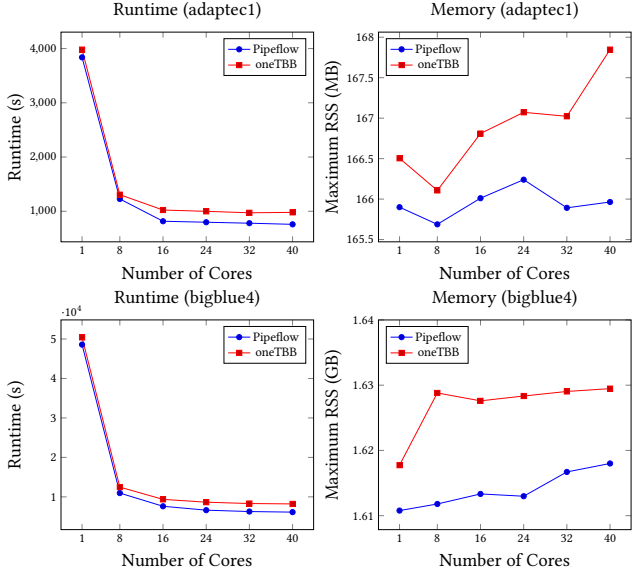